\documentclass[11pt,a4paper,authoryear]{article}

\usepackage{amsfonts,amssymb,amsthm} %
\usepackage{natbib}
\usepackage[colorlinks,citecolor=black,urlcolor=black]{hyperref}
\usepackage{amsmath} %
\usepackage{fullpage}
\usepackage{enumerate}

\newcommand{\argmax}{\operatornamewithlimits{argmax}}

\theoremstyle{plain}
\newtheorem{theorem}{Theorem}%[section]
\newtheorem{proposition}{Proposition}
\newtheorem{problem}{Problem}
\theoremstyle{definition} %%

\newcommand{\E}{{\mathbb{E}}}

\newcommand{\e}{{\bf e}}
\newcommand{\dd}{{\rm d}}

\title{\Large \bf Optimal contracts under competition when uncertainty from
  adverse selection and moral hazard are present\footnotetext{Natalie
    Packham, Berlin School of Economics and Law, Department of
    Business and Economics, Badensche Str.\ 52,
  10825 Berlin, Germany. Email: packham@hwr-berlin.de}%
\footnotetext{Accepted for publication in {\em Statistics and Probability Letters}}}
\author{\large N.\ Packham}

\begin{document}
\maketitle

\begin{abstract}
  In a continuous-time setting where a risk-averse agent controls the
  drift of an output process driven by a Brownian motion, optimal
  contracts are linear in the terminal output; this result is
  well-known in a setting with moral hazard and -- under stronger
  assumptions -- adverse selection. 
  We show that this result continues to hold when in addition
  reservation utilities are type-dependent. This type of problem
  occurs in the study of optimal compensation problems involving
  competing principals.\medskip

\end{abstract}

\maketitle

\section{Introduction}
\label{sec:introduction}

The purpose of this note is to show that linear payoffs arise as
optimal contracts offered by a principal to agents in a setting with
moral hazard, private information and competition for agents.  The
design of compensation schemes in a principal-agent relationship
typically involves a bonus component to incentivize the agent (e.g.\
employee, worker) to act in the principal's (employer's) interest. The
complicated, often highly non-linear form of contracts suggested by
economic models is seldom met in reality where bonus schemes
encountered are simple and frequently linear in the output generated
by the agent. One explanation for this gap between theory and practice
lies in the fact that linear compensation schemes are robust in the
richer, more diffuse real environment, whereas the optimal contract in
a highly stylised economic model fails at the slightest change of
model assumptions or parameters.

Against this backdrop, \citet{Holmstrom1987} show that the optimal
contract for an agent controlling the drift of an output process
driven by a Brownian motion, when the principal observes only the
output process, is a linear function of the terminal output. The
principal's inability to directly observe the action of the agent is a
{\em moral hazard problem}, implying that the agent's effort cannot be
contractually agreed, but instead the agent must be incentivized to
exert effort.  Restricting the principal's observability to the
terminal output, \citet{Sung2005} extends this to a setting where
agents have different capabilities that the principal cannot observe
or otherwise directly infer ({\em private information}). In this {\em
  adverse selection problem\/}, it is optimal to offer a menu of
linear contracts designed for the different agent types, incentivizing
each agent to choose the contract designed for their type ({\em
  screening}).

Screening typically provides agent types with different utilities. As
a consequence, if there are several principals competing for agents,
reservation utilities for agents will be type-dependent (as opposed to
the case, for example, where the outside option is not to work at
all). This setting is found in the recent literature on competition,
e.g.\ \citet{Jullien2000,Benabou2016,Bannier2016}.  We prove that the
linearity of contracts in the Holmstrom-Milgrom model carries over to
the case of moral hazard, private information and competition.  The
proofs use techniques from stochastic control theory, see e.g.\
\citet{Fleming1975}.

\section{The model}
\label{sec:model}

The setting is similar to \citet{Holmstrom1987},
\citet{Schaettler1993} and \citet{Sung2005}. An agent employed by a
principal exerts (costly) effort, which in turn increases her
output. The principal observes only the output, consisting of the
agent's effort level and random noise ({\em moral hazard}). Hence, the
effort level cannot be contracted and must be incentivized by a bonus
scheme. With agents of different capabilities that are unobservable 
by the principal ({\em private information}), the principal may find
it optimal to offer different contracts specifically designed for each
type, and must take care that each contract appeals most to the type
that it is designed for. Finally, with several principals competing
for agents, reservation utilities are derived from the outside option
of being employed by another principal. Just as the utilities offered
to agents of different capabilities differ, so do their reservation 
utilities. These must be taken into account by a principal wishing to
employ agents of all types. The principal is hence faced with an
optimization problem taking into account random noise, 
unobservable agent types and type-dependent reservation
utilities. This type of problem can be solved using martingale theory
to derive a Hamilton-Jacobi-Bellman equation, which in turn leads to
constant controls that do not depend on time. 
First, we formalize the setting and then give the principal's
optimization problem in the next section.

 A risk-neutral principal
employs a risk-averse agent whose preferences are expressed by CARA
utility with parameter $\rho$, i.e., $U(x)=1-\e^{-\rho x}$. The
agent's capabilities $\theta\in \{\theta_L,\theta_H\}$ are private
information. The probabilities $\alpha$, resp.\ $1-\alpha$, of meeting
an $H$-type, resp.\ $L$-type, agent are known to the principal.

A type-$k$ agent, by exerting effort $\mu=(\mu_t)_{t\geq 0}$, which is
assumed to be bounded, controls the drift of an output process with
dynamics 
\begin{equation*}
  \dd Z_t = \mu_t\,\theta_k\,\dd t + \sigma\, \dd W_t, \quad t\geq 0,
\end{equation*}
with $\sigma>0$ and where $W=(W_t)_{t\geq 0}$ is a Brownian motion
independent of $\theta$. The agent observes the Brownian motion $W$,
so that $\mu$ is adapted to the filtration generated by the Brownian
motion. The agent's effort is subject to an instantaneous cost
$c(\mu_t)$, $t\geq 0$, where $c(0)=0$, $c$ is strictly increasing,
convex and continuously differentiable. In addition, we require
$c'''\geq 0$ for Proposition \ref{prop:optimalcontract}. We shall
assume that $\theta_H>\theta_L$, expressing that a type-$H$ agent
generates a higher drift at equal effort cost than a type-$L$ agent.

The principal observes the output process $Z=(Z_t)_{t\geq 0}$, but
neither $W$ nor $\mu$, so the compensation, realised at time $1$, can
be contingent on $(Z_t)_{0\leq t\leq 1}$ only.  At time $0$, the
principal offers contracts, consisting of sharing rules
$\{S((Z_t)_{t\in[0,1]},\theta_k), k\in \{H,L\}\}$,
where $S(\cdot,\theta_k)$ denotes the sharing rule {\em designed\/} for
a type-$k$ agent. An agent choosing contract $S(\cdot,\theta_k)$
at time $0$, receives at time $1$
\begin{equation*}
  S((Z_t)_{0\leq t\leq 1},\theta_k) - \int_0^1 c(\mu_t)\, \dd t,
\end{equation*}
while the principal receives
\begin{equation*}
  Z_1 - S((Z_t)_{0\leq t\leq 1},\theta_k).
\end{equation*}

\section{The principal's problem}
\label{sec:principals-problem}

Denote by $S(\cdot,\theta_m)$ the contract designed for an $m$-type
agent. An agent of type $k\in\{H,L\}$, when choosing contract
$S(\cdot, \theta_m)$ exerts effort $\mu^{k,m}=(\mu_t^{k,m})_{t\geq 0}$
and derives certainty equivalent $w_{k,m}$ at time $1$. Whenever
$k=m$, we write $\mu^k$ and $w_k$.  The principal's problem is as
follows:

\begin{problem}
  \label{problem}
  Choose controls $\{\mu^H, \mu^L\}$ and a menu of contracts
  $\{S(\cdot,\theta_H)$, $S(\cdot,\theta_L)\}$ maximising
  \begin{equation*}
    \alpha \, \E\left[Z_1^H - S((Z_t^H)_{0\leq t\leq 1},
      \theta_H)\right] %
    + (1-\alpha) \,\E\left[ Z_1^L - S((Z_t^L)_{0\leq t\leq 1}, \theta_L)\right],
  \end{equation*}
  subject to \smallskip
  \begin{enumerate}[(1)]
    \addtolength{\itemsep}{5pt}
  \item $\dd Z_t^k = \mu_t^k\theta_k\, \dd t + \sigma \dd W_t$,
    $k\in \{H,L\}$,
  \item
    $\mu_t^{k,m}\in \argmax_{(\mu_t)_{0\leq t\leq 1}}
    \E\left[U\left(S((Z_t)_{0\leq t\leq 1},\theta_m) - \int_0^1
        c(\mu_t)\, \dd t\right)\right]$, where
    $\dd Z_t=\mu_t\theta_k \, \dd t + \sigma\, \dd W_t$, and
    $k,m\in \{H,L\}$,
  \item
    $\E\left[U\left(S((Z_t^k)_{0\leq t\leq }, \theta_k) - \int_0^1
        c(\mu_t^k)\, \dd t\right)\right]\geq
    \E\left[U\left(S((Z_t^{k,m})_{0\leq t\leq 1},\theta_m) - \int_0^1
        c(\mu_t^{m,k})\, \dd t\right)\right]$, where
    $\dd Z_t^{k,m} = \mu_t^{k,m}\theta_k\, \dd t + \sigma\, \dd W_t$
    and $m,k\in \{H,L\}$,\hfill {\em (ICC)}
  \item
    $\E\left[U\left(S((Z_t^k)_{0\leq t\leq }, \theta_k) - \int_0^1
        c(\mu_t^k)\, \dd t\right)\right]\geq U(w_k)$, where
    $k\in \{H,L\}$, \hfill {\em (PC)}.
  \end{enumerate}
\end{problem}
\medskip

The first constraint defines the dynamics of the output processes of
each type when choosing the contract designed for her. The second
constraint expresses that agents maximise their expected utility.  The
third constraint, the incentive compatibility constraint (ICC), makes
each agent optimally choose the contract designed from her. Finally,
the fourth constraint, the participation constraint (PC), ensures that
an agent contracts with the principal instead of choosing her outside
option with certainty equivalent $w_k$ (the general results do not
change if it were unprofitable to attract a particular agent type).

\section{The agent's choice of drift}

We consider the agent's problem when faced with a menu of
contracts. The following result is slightly adapted from
\citet{Holmstrom1987}.

\begin{theorem}
  \label{theorem:hm}
  The adapted stochastic process $(\mu_t)_{0\leq t\leq 1}$ is
  implemented with certainty equivalent $w$ by a type-$k$ agent
  by a sharing rule $S((Z_t)_{0\leq t\leq 1},\theta_k)$ only if
  \begin{multline}
    S((Z_t)_{0\leq t\leq 1}, \theta_k) = w + \int_0^1 c(\mu_t)\, \dd t +
    \int_0^1 \frac{c'(\mu_t)}{\theta_k}\, \dd Z_t - \int_0^1 c'(\mu_t)
    \mu_t\, \dd t + \frac{\rho}{2}\int_0^1
    \left(\frac{c'(\mu_t)}{\theta_k}\right)^2\, \sigma\, \dd t.
    \label{eq:7}
  \end{multline}
\end{theorem}
\begin{proof}
  See Theorem 6 of \citet{Holmstrom1987} and Corollary 4.1 of
  \citet{Schaettler1993}.
\end{proof}

The first two terms provide a certainty equivalent of $w$ and a direct
compensation of the effort cost should the agent choose to exert
effort $(\mu_t)_{0\leq t\leq 1}$. The third term incentivizes the
agent to choose effort level $\mu$. The last two terms compensate the
agent for the mean and risk of the output process, i.e., they
correspond to the certainty equivalent of the third term.

\begin{proposition}
  \label{prop:effort}
  A type-$k$ agent, when choosing the contract designed for the
  $m$-type agent, derives expected utility
  \begin{equation}
    \label{eq:9}
    \E\left[U\left(w_m + \int_0^1 c'(\mu^{k,m}_t)\mu_t^{k,m}-
        c(\mu_t^{k,m}) - \left(c'(\mu_t^m)\mu_t^m -
          c(\mu_t^m)\right)\, \dd t\right)\right],
  \end{equation}
  where $\mu^{k,m}$ denotes the $k$-type agent's optimal control,
  which solves
  \begin{equation}
    \label{eq:4}
    c'(\mu_t^{k,m}) = \frac{\theta_k}{\theta_m} c'(\mu_t^m). 
  \end{equation}
  Furthermore, the $H$-type agent exerts greater effort and derives
  greater utility from a given contract than the $L$-type agent.
\end{proposition}

\begin{proof}
  Let $S((Z_t)_{0\leq t\leq 1},\theta_m)$ be the contract designed for
  an $m$-type agent offering $w_m$ and implementing $\mu^m$.  A
  type-$k$ agent implementing $\mu$, that is,
  $\dd Z_t = \mu_t\theta_k\, \dd t + \sigma\, \dd W_t$, $t\geq 0$,
  derives expected utility
  \begin{multline*}
    \E\left[U\left(w_m + \int_0^1 c(\mu_t^m) - c(\mu_t)\, \dd t +
        \int_0^1 \frac{c'(\mu_t^m)}{\theta_m}\, \dd Z_t %
        - \int_0^1 c'(\mu_t^m) \mu_t^m - \frac{\rho}{2}
        \left(\frac{c'(\mu_t^m)}{\theta_m}\right)^2\sigma^2\, \dd
        t\right)\right],
  \end{multline*}
  which can be written as $\E\left[U(X_1^{\mu})\right]$ with state
  process
  \begin{multline*}
    X_u^{\mu}:=w_m\, u + \int_0^u c(\mu_t^m)-c(\mu_t)\, \dd t +
    \int_0^u \frac{c'(\mu_t^m)}{\theta_m} \sigma\, \dd W_t \\%
    + \int_0^u c'(\mu_t^m)\mu_t \frac{\theta_k}{\theta_m} -
    c'(\mu_t^m) \mu_t^m +
    \frac{\rho}{2}\left(\frac{c'(\mu_t^m)}{\theta_m}\right)^2\sigma^2\,
    \dd t.
  \end{multline*}
  By the It\^o formula,
  \begin{equation*}
    \dd X_t^{\mu} = \left[w_m + c(\mu_t^m) - c(\mu_t) +
      c'(\mu_t^m)\left(\mu_t\frac{\theta_k}{\theta_m} - \mu_t^m\right) +
      \frac{\rho}{2}
      \left(\frac{c'(\mu_t^m)}{\theta_m}\right)^2\sigma^2\right]\, \dd
    t + 
    \frac{c'(\mu_t^m)}{\theta_m}\sigma\, \dd W_t.
  \end{equation*}

  Define $J(t,x;\mu)=\E\left[U(X_1^{\mu})|X_t^{\mu}=x\right]$, which
  is once (twice) continuously differentiable in $t$ ($x$) (applying
  Dominated Convergence for differentiating inside the expectation
  operator), so that $J(0,x;\mu)=\E[U(X_1^{\mu})]$ corresponds to the
  objective function. Because $J$ is a martingale, the following PDE
  holds:
  \begin{equation*}
    J_t + \frac{1}{2}\left(\frac{c'(\mu_t^m)}{\theta_m}\right)^2
    \sigma^2 J_{xx} + \left[w_m + c(\mu_t^m) - c(\mu_t) + 
      c'(\mu_t^m)\left(\mu_t\frac{\theta_k}{\theta_m} - \mu_t^m\right) +
      \frac{\rho}{2}
      \left(\frac{c'(\mu_t^m)}{\theta_m}\right)^2\sigma^2\right]\,
    J_x=0, 
  \end{equation*}
  where $J_t, J_x, J_{xx}$ denote the respective partial first- and
  second-order derivatives. Setting $V(t,x)=\sup_{\mu} J(t,x;\mu)$,
  the Hamilton-Jacobi-Bellman PDE is
  \begin{equation*}
    V_t + \frac{1}{2}\sigma^2 V_{xx} + \sup_{\mu}
    \left[c'(\mu_t^m) \frac{\theta_k}{\theta_m}\mu_t -
      c(\mu_t)\right] V_x + \left[c(\mu_t^m) -
      c'(\mu_t^m) \mu_t^m + \frac{\rho}{2}
      \left(\frac{c'(\mu_t^m)}{\theta_m}\right)^2 \sigma^2\right]
    V_x=0,  
  \end{equation*}
  with boundary condition $V(1,x)=U(x)$.  The agent's Hamiltonian is
  given by
  \begin{equation*}
    \mathcal H=-c(\mu_t) +
    c'(\mu_t^m)\mu_t \frac{\theta_k}{\theta_m},
  \end{equation*}
  leading to the optimal effort choice $\mu_t^{k,m}$ fulfilling the
  FOC \eqref{eq:4}.

  If $k=H$ and $m=L$, the FOC expresses that the $H$-type agent exerts
  greater marginal effort on the $L$-type's contract than the
  $L$-type. Moreover, because of the convexity of $c$, it follows that
  $\mu_t^{H,L}>\mu_t^L$. Conversely, if $k=L$ and $m=H$, then
  $\mu_t^{L,H}<\mu_t^H$.

  The $k$-type agent's expected utility is given by \eqref{eq:9}.
  Because $c$ is strictly increasing and convex, the mean value
  theorem implies $c'(x)x > c(x)$, and
  \begin{equation*}
    \frac{\dd}{\dd x} \left[c'(x) x - c(x)\right] =
    c''(x) x>0, \quad x>0. 
  \end{equation*}
  Therefore, when $m=H$ and $k=L$, the integral is strictly positive,
  and the $H$-type agent derives a greater utility from the $L$-type's
  contract than the $L$-type. Conversely, if $m=L$ and $k=H$, the
  $L$-type derives a smaller utility from the $H$-type's contract than
  the $H$-type.
\end{proof}

The classical result is obtained that if both types' reservation
utilities are equal, then the $H$-type has an incentive to imitate if
second-best contracts were offered (i.e., contracts with moral hazard
only), in which case the contract designed for the $L$-type needs to
be distorted to prevent the $H$-type from imitating (e.g.\
\citet{Salanie2005}). If reservation utilities are type-dependent,
then the situation may be reversed, and the $H$-type's contract may
need to be distorted to prevent the $L$-type from imitating.

\section{Optimal contracts}

Turning to Problem \ref{problem}, we restrict the analysis to the case
where the $H$-type has an imitation incentive; the case when the
difference of the types' reservation utilities is sufficiently large
for the $L$-type to have an imitation incentive is treated in a
similar way. We omit the proof of the following well-known results
when the $H$-type has an imitation incentive
(e.g. \citet{Salanie2005}): The $L$-type's {\em (PC)\/} is binding
(constraint (4) in Problem \ref{problem}; to attract the $L$-type) and
the $H$-type's {\em (ICC)\/} is binding (constraint (3) in Problem
\ref{problem}; to prevent the high type from imitating), while the
$L$-type's {\em (ICC)\/} is non-binding. The $H$-type's contract
features the second-best (constant) drift rate $\mu^{H,\star}$ (there
is no reason to deviate from the optimum), while the effort level
$\mu^{L,\star}$ in the contract for the $L$-type is distorted to
prevent the $H$-type from imitating.\medskip

The principal thus solves
\begin{equation*}
  \sup_{\mu^H, \mu^L, S(\cdot, \theta_H), S(\cdot, \theta_L)}\,
  \E\left[\alpha (Z_1^H - S((Z_t^H)_{0\leq t\leq 1}, \theta^H)) +
    (1-\alpha) (Z_1^L - S((Z_t^L)_{0\leq t\leq 1}, \theta^L))\right],
\end{equation*}
subject to\smallskip
\begin{enumerate}[(1)]
  \addtolength{\itemsep}{3pt}
\item
  $\E\left[U\left(S((Z_t^L)_{0\leq t\leq 1}, \theta_L) - \int_0^1
      c(\mu_t^L)\, \dd t\right)\right]=U(w_L)$;\hfill {\em (PCL)}
\item $\left.\right.$\vspace*{-3\baselineskip}
    
    \begin{multline*}
      \E\left[U\left(S((Z_t^{H,L})_{0\leq t\leq 1}, \theta^L) -
          \int_0^t c(\mu_t^{H,L})\, \dd t\right)\right] %
      = \E\left[U\left(S((Z_t^H)_{0\leq t\leq 1}, \theta^H) -
          c(\mu^{H,\star})\right)\right] \\%
      = \E\Big[1-\exp\Big(-\rho\Big(\underbrace{w_L + \int_0^1
        c'(\mu_t^{H,L})\mu_t^{H,L} - c(\mu_t^{H,L}) -
        \left[c'(\mu_t^L)\mu_t^L - c(\mu_t^L)\right]\, \dd
        t}_{=:u_1(\mu^L)}\Big)\Big)\Big],
    \end{multline*}
    with $\mu_t^{H,L}$ given by Equation \eqref{eq:4}; \hfill {\em
      (ICCH)}
  \item
    $\E\left[U\left(S((Z_t^H)_{0\leq t\leq 1}, \theta^H) -
        c(\mu^{H,\star})\right)\right] \geq \E[U(w_H)]$;
    \hfill {\em (PCH)}
  \end{enumerate}

  \begin{proposition}
    \label{prop:optimalcontract}
    Under adverse selection, moral hazard and when reservation
    utilities are type-dependent, the optimal effort levels
    $\mu^{H,\star}$, $\mu^{L,\star}$ in the contracts designed for the
    $H$-type, resp.\ $L$-type agent are constant, and optimal
    contracts are linear in the terminal outputs $Z_1^H$ and $Z_1^L$.
  \end{proposition}

  In the proof it is shown that if {\em (PCH)\/} is non-binding, then
  $\mu^{L,\star}$ satisfies
  \begin{equation}
    \label{eq:5}
    c'(\mu^{L,\star}) = \left(\theta_L - \frac{\alpha}{1-\alpha} \left\{
        c''(\mu^{H,L,\star}) \mu^{H,L,\star}\frac{\partial}{\partial \mu^{L,\star}}
        \mu^{H,L,\star} - c''(\mu^{L,\star})\mu^{L,\star}\right\} \right)\,
    \left(1+\frac{\rho\, \sigma^2}{\theta_L^2}
      c''(\mu^{L,\star})\right)^{-1}, 
  \end{equation}
  By the assumption that $c'''\geq 0$ it follows directly that the
  optimal control is smaller than the second best optimal control
  without adverse selection (which is obtained when $\alpha=0$).  If
  {\em (PCH)\/} is binding, then $\mu^{L,\star}$ satisfies
  \begin{equation}
    \label{eq:2}
    w_H - w_L = c'(\mu^{H,L,\star}) \mu^{H,L,\star} - c(\mu^{H,L,\star})
    - [c'(\mu^{L,\star}) \mu^{L,\star}- c(\mu^{L,\star})]. 
  \end{equation}

\begin{proof}
  First, assume that {\em (PCH)\/} is non-binding (this case arises
  when the difference between reservation utilities is small or zero).
  From {\em (ICCH)\/}, where the right-hand side depends only on the
  effort level incentivized by the principal, it follows that the
  $H$-type's certainty equivalent is
  $\displaystyle -\frac{\ln \E[\e^{-\rho\, u_1(\mu^L)}]}{\rho}$.

  Using Theorem \ref{theorem:hm}, the conditions imply
  \begin{multline}
    S((Z_t^L)_{0\leq t\leq 1}, \theta_L) \\ %
    = w_L + \int_0^1 c(\mu_t^L)\,\dd t + \int_0^1
    \frac{c'(\mu_t^L)}{\theta_L}\, \dd Z_t^L - \int_0^1 c'(\mu_t^L)\,
    \mu_t^L\, \dd t + \frac{\rho}{2} \int_0^1
    \left(\frac{c'(\mu_t^L)}{\theta_L}\right)^2 \sigma^2\, \dd
    t\label{eq:8}
  \end{multline}
  and
  \begin{multline}
    % \label{eq:3}
    S((Z_t^H)_{0\leq 1\leq t}, \theta_H) \\ %
    = \frac{-\ln \E[\e^{-\rho\, u_1(\mu^L)}]}{\rho} + c(\mu^{H,\star})
    + \frac{c'(\mu^{H,\star})}{\theta_H} Z_1^H - c'(\mu^{H,\star})
    \mu^{H,\star} + \frac{\rho}{2} \left(\frac{c'(\mu^{H,\star})}
      {\theta_H}\right)^2 \sigma^2. \label{eq:1}
  \end{multline}
  Setting
  \begin{multline*}
    X_t^{\mu^L}:=\alpha \left(Z_t^H -\left\{ -\frac{\ln\E[\e^{-\rho\,
            u_t(\mu^L)}]}{\rho} + c(\mu^{H,\star}) t + \frac{\rho}{2}
        \left(\frac{c'(\mu^{H,\star})}{\theta_H}\right)^2 \sigma^2\,
        t+ \frac{c'(\mu^{H,\star})}{\theta_H}\sigma\,
        W_t\right\}\right) \\
    + (1-\alpha) \left(Z_t^L - \left\{w_L\, t + \int_0^t c(\mu_u^L)\,
        \dd u + \frac{\rho}{2} \int_0^t
        \left(\frac{c'(\mu_u^L)}{\theta_L}\right)^2\sigma^2\, \dd u +
        \int_0^t \frac{c'(\mu_t^L)}{\theta_L}\, \sigma\, \dd
        W_t\right\}\right),
  \end{multline*}
  with $\dd Z_t^L = \mu_t^L\, \theta_L\, \dd t + \sigma\, \dd W_t$ and
  $\dd Z_t^H = \mu^{H,\star}\theta_H\, \dd t + \sigma\, \dd W_t$, the
  principal's problem is
  \begin{equation*}
    \sup_{\mu_L} \E[X_1^{\mu^L}]. 
  \end{equation*}
  The dynamics of $X^{\mu^L}$ are
  \begin{align*}
    \dd X_t^{\mu^L} %
    &= \alpha \Bigg( \left\{\mu^{H,\star} \theta_H +
      \frac{\partial}{\partial t} \frac{\ln 
      \E[\e^{-\rho\, u_t(\mu^L)}]}{\rho} - c(\mu^{H,\star}) -\frac{\rho}{2}
      \left(\frac{c'(\mu^{H,\star})}{\theta_H}
      \right)^2 \sigma^2\right\}\, \dd t\Bigg. \\
    &\phantom{=\alpha\,} +
      \left(1-\frac{c'(\mu^{H,\star})}{\theta_H}\right)\sigma\, \dd W_t
      \Bigg.\Bigg) 
    \\ %
    &+ (1-\alpha) \left(\left\{\mu_t^L\theta_L -
      w_L - c(\mu_t^L) - \frac{\rho}{2}
      \left(\frac{c'(\mu_t^L)}{\theta_L}\right)^2\sigma^2\right\}\,
      \dd t + \left(1-\frac{c'(\mu_t^L)}{\theta_L}\right) \sigma\, \dd
      W_t\right), 
  \end{align*}
  and the principal's Hamiltonian is
  \begin{equation*}
    \mathcal H = \alpha\, \frac{\partial}{\partial t} \frac{\ln \E[\e^{-\rho\,
        u_t(\mu^L)}]}{\rho} + (1-\alpha)\left\{ \mu_t^L \theta_L -
      c(\mu_t^L) - \frac{\rho}{2} \left(\frac{c'(\mu_t^L)} {\theta_L}
      \right)^2 \sigma^2\right\},
  \end{equation*}
  with
  \begin{equation}
    \label{eq:6}
    \frac{\partial}{\partial t} \frac{\ln \E[\e^{-\rho\,
        u_t(\mu^L)}]}{\rho} %
    = -\frac{\E\left[\e^{-\rho\, u_t(\mu^L)}\, \left\{w_L +
          c'(\mu_t^{H,L})\mu_t^{H,L}-c(\mu_t^{H,L}) - [c'(\mu_t^L)\mu_t^L -
          c(\mu_t^L)]\right\}\right]} {\E[\e^{-\rho\, u_t(\mu^L)}]}.
  \end{equation}
  Equation \eqref{eq:6} describes the change in certainty equivalent
  offered to the $H$-type agent, including the information rent to
  prevent her from imitating. This is $\mathcal F_0$-measurable, i.e.,
  fixed at time $0$. By the principle of optimality, the optimal
  change in certainty equivalent does not depend on any particular
  time $t$; likewise the optimal control does not depend on the
  particular time $t$, so that Equation \eqref{eq:6} is constant. This
  is fulfilled for a deterministic and constant control. A constant
  optimal control is necessary as well, as the Hamiltonian is
  optimised by a deterministic choice of $\mu_t^L$, which is constant
  by the principle of optimality. Hence, \eqref{eq:6} becomes
  \begin{equation*}
    \frac{\partial}{\partial t} \frac{\ln \E[\e^{-\rho\, u_t(\mu^L)}]}
    {\rho} = -\left\{w_L +
      c'(\mu^{H,L})\mu^{H,L}-c(\mu^{H,L}) - [c'(\mu^L)\mu^L -
      c(\mu^L)]\right\}, 
  \end{equation*}
  and the Hamiltonian simplifies to
  \begin{multline*}
    \mathcal H = -\alpha \Bigg\{
    \underbrace{c'(\mu^{H,L})}_{=c'(\mu^L) \theta_H/\theta_L}
    \mu^{H,L}-c(\mu^{H,L}) - [c'(\mu^L)\mu^L -  c(\mu^L)]\Bigg\} \\
    + (1-\alpha) \left\{ \mu^L \theta_L - c(\mu^L) - \frac{\rho}{2}
      \left(\frac{c'(\mu^L)} {\theta_L} \right)^2 \sigma^2\right\}.
  \end{multline*}
  The optimum is determined via
  \begin{multline*}
    \frac{\partial}{\partial \mu^L} \mathcal H = -\alpha \left\{
      \frac{\theta_H}{\theta_L} c''(\mu^L)
      % c''(\mu^{H,L})
      \mu^{H,L}\, \frac{\partial}{\partial \mu^L} \mu^{H,L}
      - c''(\mu^L)\mu^L\right\} \\
    + (1-\alpha) \left\{\theta_L - c'(\mu^L) - \frac{\rho\,
        \sigma^2}{\theta_L^2} c'(\mu^L) c''(\mu^L)\right\},
  \end{multline*}
  which is zero if $\mu^{L,\star}$ fulfills Equation \eqref{eq:5},
  requiring this is nonnegative.  It is easily verified that this is a
  minimum by $c'''\geq 0$.  The sharing rules \eqref{eq:8} and
  \eqref{eq:1} depend only on $Z_1$ instead of $(Z_t)_{0\leq t\leq 1}$
  and are linear in $Z_1$.

  \medskip

  If {\em (PCH)\/} is binding, then
  \begin{equation*}
    \E[\e^{-\rho\, u_1(\mu^L)}] = \e^{-\rho \, w_H},
  \end{equation*}
  and
  \begin{equation*}
    \dd \E[\e^{-\rho\, u_t(\mu^L)}] = \dd \e^{-\rho\, w_H\, t} = -\rho\,
    w_H\, \e^{-\rho\, w_H\, t}\, \dd t. 
  \end{equation*}
  The left-hand side is therefore an expectation of an exponential
  accruing at a constant rate. Since $u_1(\mu^L)$ is comprised of a
  constant and an integral with respect to time, by the principle of
  optimality, the exponent itself must be constant. Hence
  $\mu^{L,\star}$ is constant. Furthermore, $\mu^{L,\star}$ solves the
  binding participation constraint, which can be expressed as in
  Equation \eqref{eq:2}.
\end{proof}

\bibliographystyle{abbrvnamed} %
\bibliography{finance}

\end{document}